\documentclass[11pt,a4paper]{article}
\usepackage[margin=30mm]{geometry}

\usepackage[T1]{fontenc}
\usepackage{graphicx}
\usepackage[algoruled,linesnumbered,algo2e,vlined]{algorithm2e}
\usepackage{amsmath,amssymb}
\usepackage{url}
\usepackage{rotating}
\usepackage{comment}
\usepackage{multirow}
\usepackage{amsthm}
\theoremstyle{plain}
\newtheorem{theorem}{Theorem}

\newtheorem{lemma}[theorem]{Lemma}

\newtheorem{problem}[theorem]{Problem}

\newcommand{\emptystr}{\varepsilon}

\newcommand{\isleaf}{\mathsf{isleaf}}
\newcommand{\lchild}{\mathsf{lchild}}
\newcommand{\rchild}{\mathsf{rchild}}

\newcommand{\rmleaf}{\mathsf{rmleaf}}
\newcommand{\leafrank}{\mathsf{leafrank}}

\newcommand{\bwdsearch}{\mathsf{bwdsearch}}

\newcommand{\gslp}{\mathcal{G}}
\newcommand{\gtext}{\mathcal{T}}
\newcommand{\expand}[1]{\langle#1\rangle}
\newcommand{\gdag}{\mathcal{D}}
\newcommand{\dic}{\mathsf{R}}

\newcommand{\pa}{\mathsf{P}}
\newcommand{\da}{\mathsf{D}}
\newcommand{\dicone}{\mathsf{R}_{1}}
\newcommand{\dictwo}{\mathsf{R}_{2}}
\newcommand{\lensum}{\mathsf{G}}
\newcommand{\cbt}{\mathsf{B}}

\newcommand{\dice}{\mathsf{R_{E}}}
\newcommand{\treee}{\mathsf{T_{E}}}
\newcommand{\marke}{\mathsf{M_{E}}}

\newcommand{\monoe}{\mathsf{S}}

\newcommand{\rank}{\mathsf{rank}}
\newcommand{\select}{\mathsf{select}}

\usepackage{hyperref}
\usepackage{authblk}

\title{Space-efficient SLP Encoding for $O(\log N)$-time Random Access}

\author[1]{Akito Takasaka\thanks{takasaka.akito977@mail.kyutech.jp}}
\author[1]{Tomohiro I\thanks{tomohiro@ai.kyutech.ac.jp}}
\affil[1]{Kyushu Institute of Technology, 680-4 Kawazu, Iizuka, Fukuoka 820-8502, Japan}

\date{}

\begin{document}
\maketitle

\abstract{A Straight-Line Program (SLP) $\gslp$ for a string $\gtext$ is a context-free grammar (CFG) that derives $\gtext$ only,
which can be considered as a compressed representation of $\gtext$.
In this paper, we show how to encode $\gslp$ in $n \lceil \lg N \rceil + (n + n') \lceil \lg (n+\sigma) \rceil + 4n - 2n' + o(n)$ bits 
to support random access queries of extracting $\gtext[p..q]$ in worst-case $O(\log N + q - p)$ time, where
$N$ is the length of $\gtext$, $\sigma$ is the alphabet size, $n$ is the number of variables in $\gslp$ and
$n' \le n$ is the number of symmetric centroid paths in the DAG representation for $\gslp$.
The time complexity is almost optimal because Verbin and Yu [CPM 2013] proved that 
$O(\log N)$ term cannot be significantly improved in general with $\mathrm{poly}(n)$-space data structures.
We also present alternative encodings that achieve the same random access time with
$n \lceil \lg N \rceil + n \lceil \lg (n+\sigma) \rceil + 5n + n' + o(n)$ or $n \lceil \lg N \rceil + n \lceil \lg (n+\sigma) \rceil + 5n - n' + \sigma + o(n+\sigma)$ bits of space.
}

\section{Introduction}\label{sec:intro}
A Straight-Line Program (SLP) $\gslp$ for a string $\gtext$ is a context-free grammar (CFG) that derives $\gtext$ only.
The idea of grammar compression is to take $\gslp$ as a compressed representation of $\gtext$,
which is a useful scheme to capture repetitive substrings in $\gtext$.
In fact, the output of many practical dictionary-based compressors like RePair~\cite{Larsson1999RePair} and LZ77~\cite{1976LempelZ_ComplOfFinitSequen_TIT,1977ZivL_UniverAlgorForSequenData_IeeeTransInfTheory} 
can be considered as or efficiently converted to an SLP.\@
Moreover, SLPs have gained popularity for designing algorithms and data structures to work directly on compressed data.
For more details, see survey~\cite{2015Lohrey_GrammBasedTreeCompr} and references therein.

One of the most fundamental tasks on compressed string is to support random access without explicitly decompressing the whole string.
Let $\gtext$ be a string of length $N$ over an alphabet $\Sigma$ of size $\sigma$ and $\gslp$ be an SLP that derives $\gtext$ with $V$ being the set of variables.
For simplicity, we assume that SLPs are in the normal form such that every production rule is of the form $X \rightarrow YZ \in (V \cup \Sigma)^2$.
If we store the length of the string derived from every variable in $n \lg N$ bits in addition to the information of production rules,
it is not difficult to see that we can access $\gtext[p]$ for any position $1 \le p \le N$ in $O(h)$ time,
where $h$ is the height of the derivation tree of $\gslp$:
We can simulate the traversal from the root to the $p$-th leaf of the derivation tree of $\gslp$
while deciding if the current node contains the target leaf in its left child or not.
This simple random access algorithm is good enough if $\gslp$ is well balanced, i.e., $h = O(\log N)$,
but $h$ could be as large as $n$ in the worst case.

To solve this problem, Bille et al.~\cite{2015BilleLRSSW_RandomAccesToGrammCompr} showed that there is a data structure of $O(n \log N)$ bits of space 
that can retrieve any substring of length $\ell$ of $\gtext$ in $O(\log N + \ell)$ time.
Belazzougui et al.~\cite{2015BelazzouguiCPT_AccesRankAndSelecIn_ESA} showed that
the query time can be improved to $O(\log N + \ell/\log_{\sigma} N)$ 
by adding some other data structures of $O(n \log N)$ bits to accelerate accessing $O(\log_{\sigma} N)$ consecutive characters.
Verbin and Yu~\cite{2013VerbinY_DataStrucLowerBoundOn_CPM} studied lower bounds of random access data structures on grammar compressed strings and
showed that any data structure of $S = \mathrm{poly}(n)$ space needs $\Omega(\log^{1-\epsilon} N / \log S)$ time for any constant $\epsilon$,
and thus, the data structures of~\cite{2015BilleLRSSW_RandomAccesToGrammCompr,2015BelazzouguiCPT_AccesRankAndSelecIn_ESA} achieve almost optimal random access time in general.

Another approach is to transform a given grammar into a balanced grammar of height $O(\log N)$, and apply the above-mentioned simple random access algorithm.
Ganardi et al.~\cite{2021GanardiJL_BalanStraigLineProgr} showed that 
any SLP can be transformed in linear time into a balanced grammar without increasing its order in size.
The result was refined in~\cite{2021Ganardi_ComprByContrStraigLine_ESA} for contracting SLPs, which have a stronger balancing condition.
These results are helpful not only for random access but also for other operations that can be done depending on the height of the derivation tree.
However, the constant-factor blow-up in grammar size can be a problem in space-sensitive applications.

In order to keep space usage small in practice, it is important to devise a space-economic way to encode $\gslp$ and auxiliary data structures,
as such importance has been highlighted by increasing interest of succinct data structures.
Since the righthand side of each production rule can be stored in $2 \lceil \lg (n+\sigma) \rceil$ bits,
a naive encoding for $\gslp$ would use $2n \lceil \lg (n+\sigma) \rceil + n \lceil \lg N \rceil$ bits of space,
which is far from an information-theoretic lower bound $n \lg (n+\sigma) + 2n + o(n)$ bits for representing $\gslp$~\cite{2013TabeiTS_SuccinGrammCompr_CPM}.
Maruyama et al. proposed a succinct encoding of $n \lceil \lg (n + \sigma) \rceil + 2n + o(n)$ bits~\cite{2013MaruyamaTSS_FullyOnlinGrammCompr_SPIRE} for $\gslp$,
which can be augmented with additional $n \lceil \lg \frac{N}{n} \rceil + 2n + o(n)$ bits to support $O(h + \ell)$-time random access.
Other practical encodings for random access were studied in~\cite{2020GagieIMNSBT_PractRandomAccesToSlp_SPIRE,2024ClearyWDI_RevisFolklAlgorForRandom_SPIRE},
but its worst-case query time is $O(h + \ell)$.

In this study, we propose a novel space-efficient SLP encoding that supports random access queries in worst-case $O(\log N + \ell)$ time.
In so doing, we simplify some ideas of~\cite{2015BilleLRSSW_RandomAccesToGrammCompr} and adjust them to work with succinct data structures.
For example, we replace the heavy-path decomposition of a Directed Acyclic Graph (DAG) with 
the symmetric centroid decomposition proposed in~\cite{2021GanardiJL_BalanStraigLineProgr}.
We decompose the DAG representation of the derivation tree of $\gslp$ into disjoint Symmetric Centroid paths (SC-paths) so that
every path from the root to a leaf passes through $O(\log N)$ distinct SC-paths.
We augment each SC-path with compacted binary tries to support interval-biased search,
which leads to $O(\log N)$-time random access.
Under a standard Word-RAM model with word size $\Omega(\lg N)$ we get the following result:\footnote{The result is enhanced with encoding (II) and (III) from the preliminary version in~\cite{2024TakasakaI_SpaceEfficSlpEncodFor_SPIRE}.}
\begin{theorem}\label{theorem:enc}
  Let $\gtext$ be a string of length $N$ over an alphabet of size $\sigma$.
  An SLP $\gslp$ for $\gtext$ can be encoded in (I) $n \lceil \lg N \rceil + (n + n') \lceil \lg (n+\sigma) \rceil + 4n - 2n' + o(n)$,
  (II) $n \lceil \lg N \rceil + n \lceil \lg (n+\sigma) \rceil + 5n + n' + o(n)$, or
  (III) $n \lceil \lg N \rceil + n \lceil \lg (n+\sigma) \rceil + 5n - n' + \sigma + o(n+\sigma)$ bits of space
  while allowing to retrieve, given $1 \le p \le q \le N$, 
  the substring $\gtext[p..q]$ in $O(\log N + q - p)$ time,
  where $n$ is the number of variables of $\gslp$ and 
  $n' \le n$ is the number of SC-paths in the DAG representation for $\gslp$.
\end{theorem}

Our effort in this paper focuses on squeezing the space needed to encode SLP itself while achieving $O(\log N)$-time random access.
In particular, the space for encoding (II) and (III) of Theorem~\ref{theorem:enc} matches the information-theoretic lower bound to store SLP itself if we ignore $n \lceil \lg N \rceil$ bits used to store expansion lengths of variables and other insignificant lower terms.
Now, the most space consuming term in our encodings is $n \lceil \lg N \rceil$ bits used to store expansion lengths.
Although we are not aware of precise lower bounds of bits needed for $O(\log N)$-time random access, it would be unlikely that we can eliminate the $n \lceil \lg N \rceil$ term completely.
Still, if $N = O(\mathrm{poly}(n))$, then our effort is meaningful because it actually leads to lowering the coefficient of the leading term of $\Theta(n \log (n + \sigma))$ bits.

\section{Preliminaries}\label{sec:prelim}

\subsection{Basic Notation}
For two integers $i$ and $j$ with $i \le j$, let $[i..j]$ represent
the integer interval from $i$ to $j$, i.e. $[i..j] := \{ i, i+1, \dots, j-1, j \}$.
If $i > j$, then $[i..j]$ denotes the empty interval.
Also, let $[i..j) := [i..j-1]$ and $(i..j] := [i+1..j]$.
We use $\lg$ to denote the binary logarithm, i.e., the logarithm to the base 2.

Let $\Sigma$ be a finite \emph{alphabet}.
An element of $\Sigma^*$ is called a \emph{string} over $\Sigma$.
The length of a string $w$ is denoted by $|w|$. 
The empty string $\emptystr$ is the string of length 0,
that is, $|\emptystr| = 0$.
Let $\Sigma^+ = \Sigma^* - \{\emptystr\}$.
The concatenation of two strings $x$ and $y$ is denoted by $x \cdot y$ or simply $xy$.
When a string $w$ is represented by the concatenation of strings $x$, $y$ and $z$ (i.e. $w = xyz$), 
then $x$, $y$ and $z$ are called a \emph{prefix}, \emph{substring}, and \emph{suffix} of $w$, respectively.
A substring $x$ of $w$ is called \emph{proper} if $|x| < |w|$.

The $i$-th character of a string $w$ is denoted by $w[i]$ for $1 \leq i \leq |w|$,
and the substring of a string $w$ that begins at position $i$ and
ends at position $j$ is denoted by $w[i..j]$ for $1 \leq i \leq j \leq |w|$,
i.e., $w[i..j] = w[i]w[i+1] \cdots w[j]$.
For convenience, let $w[i..j] = \emptystr$ if $j < i$.

\subsection{Straight-Line Programs (SLPs)}
Let $\gtext$ be a string of length $N$ over $\Sigma$.
A \emph{Straight-Line Program (SLP)} $\gslp$ for $\gtext$ is a context-free grammar that derives $\gtext$ only.
Let $V$ be the \emph{variables} (non-terminals) of $\gslp$.
We use a term \emph{symbol} to refer to an element in $(V \cup \Sigma)$.
We obtain $\gtext$ by recursively replacing the starting variable of $\gslp$
according to the production rules of $\gslp$ until every variable is turned into a sequence of \emph{characters} (terminals).
To derive $\gtext$ uniquely, the derivation process of $\gslp$ should be deterministic and end without loop.
In particular, for each variable $X$, there is exactly one production rule that has $X$ in its lefthand side,
which we call the production rule of $X$.
We denote by $\dic(X)$ the righthand side of the production rule of $X$.
For simplicity, we assume that SLPs are in the normal form such that $\dic(X) \in (V \cup \Sigma)^2$.
For any symbol $x$, let $\expand{x}$ denote the string derived from $x$, i.e.,
$\expand{x} = x$ if $x \in \Sigma$, and otherwise $\expand{x} = \expand{\dic(x)[1]}\expand{\dic(x)[2]}$.
We extend this notation so that $\expand{w} := \expand{w[1]}\expand{w[2]}\cdots\expand{w[|w|]}$ for a string $w$ over $(V \cup \Sigma)^*$.
For a symbol $x$, the \emph{derivation tree} of $x$ is the rooted tree that represents the derivation process from $x$ to $\expand{x}$.
The derivation tree of the starting symbol is called the derivation tree of $\gslp$.
Note that the derivation tree of $\gslp$ can be represented by a Directed Acyclic Graph (DAG) with $n + \sigma$ nodes and $2n$ edges, 
which we call the DAG of $\gslp$ and denote by $\gdag_{\gslp}$.
Since there is a natural one-to-one correspondence from symbols to nodes of $\gdag_{\gslp}$, 
we will sometimes use them interchangeably.

We set some assumptions under which we study the space needed to store $\gslp$.
We assume that symbols are identified by integers 
with $V = [1..n]$ and $\Sigma = [n+1..n+\sigma]$
so that each of them is represented in $\lceil \lg (n + \sigma) \rceil$ bits, 
where $n := |V|$ and $\sigma := |\Sigma|$.
For a terminal symbol associated with an integer $i \in [n+1..n+\sigma]$,
its original code on computer is assumed to be obtained easily,
e.g., by computing $i - n$ or storing a mapping table whose space usage is excluded from our space complexity.

\subsection{Succinct Data Structures}
For a bit string $B \in \{ 0, 1 \}^{*}$, we consider the following queries:
\begin{itemize}
  \item For any $b \in \{ 0, 1\}$ and $i \in [1..|B|]$, $\rank_{b}(B, i)$ returns the number of occurrences of $b$ in $B[1..i]$.
    For convenience, we let $\rank_{b}(B, i)$ return $0$ if $i < 1$, and $|B|$ if $i > |B|$.
  \item For any $b \in \{ 0, 1\}$ and $j \in [1..\rank_{b}(B, |B|)]$, $\select_{b}(B, j)$ returns the position $i$ such that $\rank_{b}(B, i) = j$ and $B[i] = b$.
    For convenience, we let $\select_{b}(B, j)$ return $0$ if $j \le 0$, and $|B|+1$ if $j > \rank_{b}(B, |B|)$.
\end{itemize}

We use the following succinct data structure on bit strings:
\begin{lemma}[\cite{2007RamanRS_SuccinIndexDictionWithApplic}]\label{lem:rrr}
  For a bit string $B \in \{ 0, 1 \}^{n}$, there is a data structure of $n + o(n)$ bits
  that supports $\rank$ and $\select$ queries in $O(1)$ time.
\end{lemma}

We also consider succinct data structures for rooted ordered full binary trees in which all internal nodes have exactly two children of left and right.
Notice that a full binary tree with $m$ nodes can be represented in $m$ bits instead of $2m$ bits needed in the case of arbitrary rooted ordered trees,
and there is a data structure to support various operations in $m + o(m)$ bits.
Although there could be several ways to encode the topology of the tree,
we employ the post-order encoding $B$ in which every node is identified by its post-order rank.
In this paper, we use the queries listed in the following lemma,
which is not new as it has been used in the literature (e.g.,~\cite{2013MaruyamaTSS_FullyOnlinGrammCompr_SPIRE,2017TakabatakeIS_SpaceOptimGrammCompr_ESA}).
We give its proof for the sake of completeness.

\begin{lemma}\label{lem:poe}
  For a full binary tree with $m$ nodes, there is a data structure of $m + o(m)$ bits to support the following queries in $O(1)$ time for a node $v$,
  where every node involved in the queries is identified by its post-order rank.
  \begin{itemize}
    \item $\isleaf(v)$ returns if $v$ is a leaf node.
    \item $\lchild(v)$ returns the left child of $v$.
    \item $\rchild(v)$ returns the right child of $v$.
    \item $\rmleaf(v)$ returns the rightmost leaf in the subtree rooted at $v$.
    \item $\leafrank(v)$ returns the number of leaves up to and including a leaf $v$.
  \end{itemize}
\end{lemma}
\begin{proof}
  We store a bit string $B[1..m]$ of length $m$ such that the $i$-th bit is $0$ if and only if 
  the node with post-order rank $i$ is a leaf node.
  We store and augment the post-order encoding $B$ with rank/select data structures of Lemma~\ref{lem:rrr} using extra $o(m)$ bits.
  We also conceptually define the so-called excess array $E[1..m]$ such that $E[i] := \rank_0(B, i) - \rank_1(B, i)$ for any $1 \le i \le m$.
  Since it holds that $E[i] > 0$ and $|E[i] - E[i-1]| = 1$, we can use the succinct data structure
  for a balanced parentheses sequence~\cite{2014NavarroS_FullyFunctStaticAndDynam}.
  In particular, we augment $B$ to compute in constant time $\bwdsearch(i, d)$ 
  that returns the maximum position $j$ with $j \le i$ and $E[j] = E[i] + d$.
  We will use $\bwdsearch(i, d)$ queries only for $d = -1$.

  Now, we show how to respond to the queries listed in our lemma.
  For $\isleaf(v)$, we return true if $B[v] = 0$ and false otherwise.
  Since the right child of $v$ immediately precedes $v$ in post order, 
  we just return $v - 1$ for $\rchild(v)$.
  The rightmost leaf in the subtree rooted at $v$ is at the maximum position $j$ 
  such that $j \le v$ and $B[j] = 0$, which can be computed by $\select_0(B, \rank_0(B, v))$.
  For $\lchild(v)$, we look for the rightmost node to the left from $v$ in $B$ that is not in the subtree rooted at the right child of $v$.
  This can be computed by $\bwdsearch(v-1, -1)$ as the subtree rooted at the right child of $v$ is also a full binary tree 
  in which the leaves exceed the internal nodes by exactly one in number.
  For a leaf $v$, we can compute $\leafrank(v)$ by $\rank_{0}(B, v)$.
\end{proof}

\subsection{Symmetric Centroid Decomposition}\label{sec:scd}
Here we briefly review the symmetric centroid decomposition proposed in~\cite{2021GanardiJL_BalanStraigLineProgr}.
Let us work on the DAG of $\gslp$ for a string of length $N$.
For every node $v$, we consider the pair $(\lfloor \lg \dot{v} \rfloor, \lfloor \lg \ddot{v} \rfloor)$ of values, where 
$\dot{v}$ and $\ddot{v}$ are the numbers of paths from the root to $v$ and from $v$ to the leaves, respectively.
Note that $\dot{v}$ and $\ddot{v}$ are both upper bounded by $N$.
An edge from $u$ to $v$ is called an \emph{SC-edge} if and only if 
$(\lfloor \lg \dot{u} \rfloor, \lfloor \lg \ddot{u} \rfloor) = (\lfloor \lg \dot{v} \rfloor, \lfloor \lg \ddot{v} \rfloor)$.
By Lemma 2.1 of~\cite{2021GanardiJL_BalanStraigLineProgr}, every node has at most one outgoing SC-edge and at most one incoming SC-edge.
Hence, a maximal subgraph connected by SC-edges forms a path, which we call an \emph{SC-path}.
Note that SC-paths include an empty path, which consists of a single node.
Lemma 2.1 of~\cite{2021GanardiJL_BalanStraigLineProgr} also states that every path from the root to a leaf 
contains at most $2 \lg N$ non-SC-edges.

\section{New Encoding of SLPs}
In this section we prove Theorem~\ref{theorem:enc}.
In what follows, $\gtext, \sigma, N, \gslp, n$ and $n'$ are used as defined in the theorem.

\subsection{Strategy to Achieve $O(\log N)$-time Random Access}\label{sec:overview}
In this subsection, we give a high-level strategy to achieve $O(\log N)$-time random access.

Given a position $p$ on $\gtext$, we simulate on $\gdag_{\gslp}$ the traversal from the root to the $p$-th leaf of the derivation tree of $\gslp$.
As described in Subsection~\ref{sec:scd}, the traversal contains at most $2 \lg N$ non-SC-edges, 
say $(x_1^{\mathsf{out}}, x_2^{\mathsf{in}})$, $(x_2^{\mathsf{out}}, x_3^{\mathsf{in}})$, $\dots$, $(x_{e}^{\mathsf{out}}, x_{e+1}^{\mathsf{in}})$ with $e \le 2 \lg N$,
where $x_{i}^{\mathsf{in}}$ and $x_{i}^{\mathsf{out}}$ are variables on the same SC-path for all $1 \le i \le e$
(for convenience, let $x_1^{\mathsf{in}}$ be the root node of $\gdag_{\gslp}$, i.e., the starting variable)
and $x_{e+1}^{\mathsf{in}} = \gtext[p]$.
See Fig.~\ref{fig:hls} for an illustration.
If we can move from $x_{i}^{\mathsf{in}}$ to $x_{i+1}^{\mathsf{in}}$ in a way of length-weighted biased search
in $O(1 + \log |\expand{x_{i}^{\mathsf{in}}}| - \log |\expand{x_{i+1}^{\mathsf{in}}}|)$ time,
$O(\log N)$-time random access can be achieved because
$O(\sum_{i = 1}^{e} (1 + \log |\expand{x_{i}^{\mathsf{in}}}| - \log |\expand{x_{i+1}^{\mathsf{in}}}|)) = O(\log N + \log |\expand{x_1^{\mathsf{in}}}| - \log |\expand{x_{e+1}^{\mathsf{in}}}|) = O(\log N)$.

\begin{figure}[t]
 \centering{
   \includegraphics[scale=0.4]{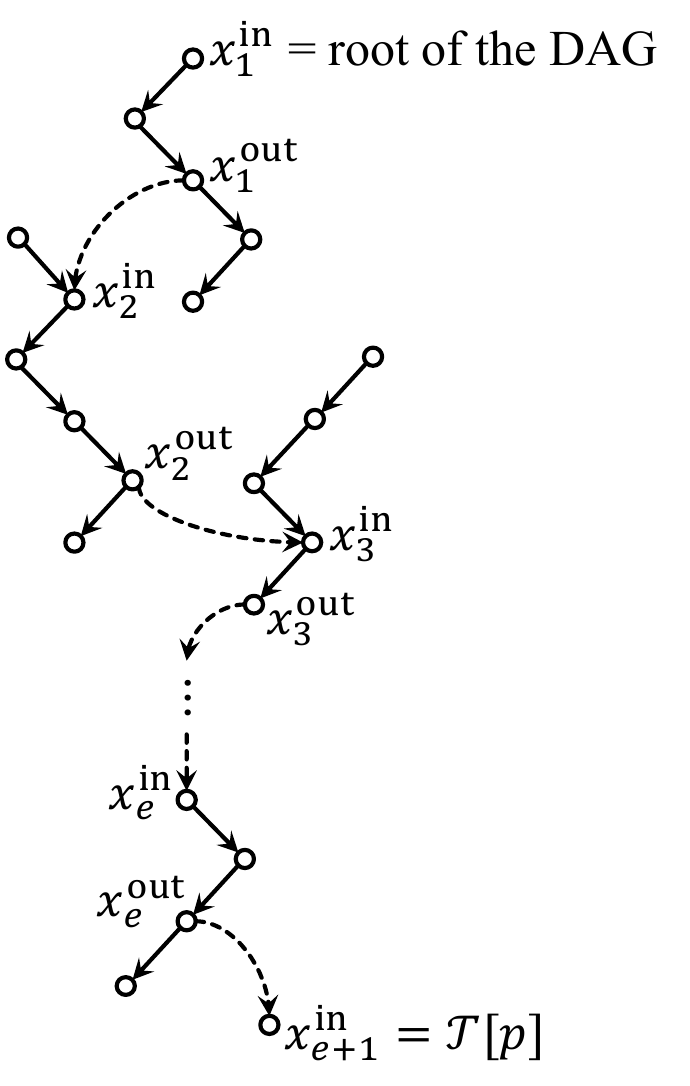}
 }
 \caption{Illustration for the high-level strategy to achieve $O(\log N)$-time random access.
   The path from the root to the target leaf $x_{e+1}^{\mathsf{in}} = \gtext[p]$ contains $e~(\le 2 \lg N)$ non-SC-edges
   $(x_1^{\mathsf{out}}, x_2^{\mathsf{in}})$, $(x_2^{\mathsf{out}}, x_3^{\mathsf{in}})$, $\dots$ and $(x_{e}^{\mathsf{out}}, x_{e+1}^{\mathsf{in}})$
   depicted by dashed arrows.
   The components connected by plain arrows are SC-paths.
   Our sub-goal is to move from $x_{i}^{\mathsf{in}}$ to $x_{i+1}^{\mathsf{in}}$ efficiently 
   in $O(1 + \log |\expand{x_{i}^{\mathsf{in}}}| - \log |\expand{x_{i+1}^{\mathsf{in}}}|)$ time.
 }
 \label{fig:hls}
\end{figure}

For an SC-path $(u_1, u_2, \dots, u_m)$ with $m$ nodes, 
there are $m+1$ non-SC-edges branching out from the SC-path,
which splits $\expand{u_1}$ into $m+1$ substrings.
A subproblem in question is to find the non-SC-edge $(u_j, v)$ such that 
$v$ contains the target position in time $O(1 + \log |\expand{u_1}| - \log |\expand{v}|)$.
In Subsection~\ref{sec:ibst}, we take this subproblem as a general interval search problem and show how to solve it.

\subsection{Compacted Binary Tries for Interval-Biased Search}\label{sec:ibst}
In this subsection, we consider the following problem.
\begin{problem}[Interval Search Problem]\label{isp}
  Preprocess a sequence $g_1, g_2, \dots, g_m$ of integers such that $g_0 = 0 < g_1 < g_2 < \dots < g_m$,
  to support interval search queries that ask, given an integer $0 < p \le g_m$, to compute $k$ with $p \in (g_k..g_{k+1}]$.
\end{problem}
Data structures for this problem have been extensively studied in the context of the predecessor search problem (e.g. y-fast trie~\cite{Willard1983LWR}).

Bille et al.~\cite{2015BilleLRSSW_RandomAccesToGrammCompr} proposed the interval-biased search tree
to answer the interval search query in $O(1 + \log g_m - \log(g_{i+1} - g_{i}))$ time.
The interval-biased search tree for $m$ intervals $(g_0..g_1], (g_1..g_2], \dots, (g_{m-1}..g_{m}]$ is a binary tree with $m$ nodes defined as follows:
The root is set to be the interval $(g_{m'}..g_{m'+1}]$ that contains the position $(g_{m} - g_{0})/2$
and its left (resp. right) child subtree is defined recursively for
$(g_0..g_1], \dots ,(g_{m'-1}..g_{m'}]$ (resp. $(g_{m'+1}..g_{m'+2}], \dots ,(g_{m-1}..g_{m}]$) if $m' > 0$ (resp. $m' < m-1$).

Unfortunately, it does not seem straightforward to work on the succinct tree representation of the interval-biased search tree:
Although one can observe that the node with in-order rank $i+1$ 
(defined in a reasonable way for trees in which each edge is categorized into left or right)
corresponds to interval $(g_i, g_{i+1}]$ in the interval-biased search tree,
the computation is not supported by the succinct tree representation of~\cite{2014NavarroS_FullyFunctStaticAndDynam}.
\footnote{Note that $\mathit{in\_rank}$ queries in~\cite{2014NavarroS_FullyFunctStaticAndDynam} work only for the nodes that have at least two children.}

In this paper, we instead show that the compacted binary trie (also known as compressed binary trie or Patricia trie~\cite{1968Morrison_PatricPractAlgorToRetriev}) 
for the binary representations of integers $g_1, g_2, \dots, g_m$ can be used for interval-biased search.
\begin{lemma}\label{lem:ibs}
  For Interval Search Problem, there is a data structure of $m \lceil \lg g_m \rceil + 2m + o(m)$ bits 
  to support interval search queries in $O(1 + \log g_m - \log(g_{k+1} - g_{k}))$ time,
  where $k$ is the answer of the query.
\end{lemma}
\begin{proof}
  We store $g_1, g_2, \dots, g_m$ naively in an array $A[1..m] \in [1..g_m]^{m}$ using $m \lceil \lg g_m \rceil$ bits of space.
  For any $1 \le i \le m$, let $b_i$ denote the $(\lceil \lg g_m \rceil)$-bit binary representation of $g_i$,
  and $c_i$ be the bit string that is obtained from $b_i$ by removing the most significant bits common to both $b_1$ and $b_m$.
  The compacted binary trie for $c_1, c_2, \dots, c_m$ forms a full binary tree with $m$ leaves and $m-1$ internal nodes and 
  its tree topology can be represented in $2m + o(m)$ bits by the data structure of Lemma~\ref{lem:poe}.
  Note that the $i$-th leaf $v_i$ corresponds to $c_i$ (and also $g_i = A[i]$), 
  and the interval $(g_i..g_{i+1}]$ can be uniquely assigned to the lowest common ancestor $u_i$ of $v_i$ and $v_{i+1}$.
  The crucial point for our interval-biased search is that the depth of $u_i$ in the tree is less than $\lceil \lg g_m \rceil - \lg (g_{i+1} - g_{i})$, 
  because the depth is upper bounded by the length $\ell$ of the longest common prefix between $b_{i}$ and $b_{i+1}$, 
  which must satisfy $2^{\lceil \lg g_m \rceil - \ell} > g_{i+1} - g_{i}$ to have two leaves $v_{i}$ and $v_{i+1}$ below $u_i$ in the trie.
  Hence, we can answer interval search queries in $O(1 + \log g_m - \log (g_{i+1} - g_{i}))$ time by a simple traversal from the root:
  When we arrive at an internal node $u$ (given by its post-order rank), 
  we can compute in $O(1)$ time its corresponding interval $(A[i]..A[i+1]]$ using $i = \leafrank(\rmleaf(\lchild(u)))$.
  We return $i$ if $p \in (A[i]..A[i+1]]$, and otherwise,
  we move to the left child $\lchild(u)$ (resp. right child $\rchild(u)$) if $p \le A[i]$ (resp. $A[i+1] < p$).
  A final remark is that we return $0$ if $p \le g_1 = A[1]$, i.e., $p \in (g_0..g_{1}]$, which should be checked prior to the binary search.
\end{proof}

\subsection{Our Encoding (I) of Theorem~\ref{theorem:enc}}
First, we arrange the variables of $\gslp$ so that $v = u+1$ holds for every SC-edge $(u, v)$ on $\gdag_{\gslp}$.
As a result, we can now assume that the variables on the same SC-path are given consecutive integers
while distinct SC-paths are ordered arbitrarily.
Let $n'$ be the number of SC-paths of $\gdag_{\gslp}$ excluding the ones that consist of the terminal symbols.

Our encoding stores a bit string $\pa[1..n]$ to separate $[1..n]$ into $n'$ segments of SC-paths, i.e.,
for any variable $u~(1 \le u \le n)$, $\pa[u] = 1$ if and only if $u$ is the last node of an SC-path.
We remark that $\pa$ contains $n'$ ones and $n-n'$ zeros.

Next we show how to encode the information of $\dic$.
For any $u$ with $\pa[u] = 0$, one of its children is $u+1$, and hence, we can store $\dic$ in $(n + n') \lceil \lg (n+\sigma) \rceil + n - n'$ bits
in addition to $n$ bits of $\pa$ as follows.
Let $\da[1..n-n']$ be a bit string of length $n-n'$ that indicates the direction of the non-SC-edge $(u, v)$ branching out from every node $u$ with $\pa[u] = 0$,
and $\dicone[1..n-n'] \in [1..n+\sigma]^{n-n'}$ be an integer array of length $n-n'$ that stores $v$'s.
More precisely, for any variable $u~(1 \le u \le n)$ with $\pa[u] = 0$, 
we set $\da[\rank_0(\pa, u)]$ to be $0$ if and only if the non-SC-edge $(u, v)$ branching out from $u$ is the left child,
and store $v$ in $\dicone[\rank_0(\pa, u)]$.
For the nodes $u$ with $\pa[u] = 1$, we create another integer array $\dictwo[1..2n'] \in [1..n+\sigma]^{2n'}$ such that
$\dictwo[2 \cdot \rank_1(\pa, u)-1..2 \cdot \rank_1(\pa, u)]$ store $\dic(u)$.

Using $\pa$, $\da$, $\dicone$ and $\dictwo$, we can compute in $O(1)$ time the children $\dic(u)$ of a given variable $u$
if we augment $\pa$ with the data structure of Lemma~\ref{lem:rrr} for rank/select queries.
Also, $r = \rank_1(\pa, u-1) + 1$ tells us that $u$ is on the $r$-th SC-path on $\pa$,
and the interval for the SC-path $u$ belongs to can be computed in $O(1)$ time by $[\select_1(\pa, r-1)+1..\select_1(\pa, r)]$.

Finally, we add $(n \lceil \lg N \rceil + 2n - n' + o(n))$-bits data structure to achieve $O(\log N)$-time random access.
From now on, we focus on the $r$-th SC-path.
Let $m := \select_1(\pa, r) - \select_1(\pa, r-1)$ and $u_i := \select_1(\pa, r-1) + i$ for any $0 \le i \le m$, 
then $(u_1, u_2, \dots, u_m)$ is the sequence of nodes on the $r$-th SC-path.
Let $t$ be the number of non-SC-edges $(u_i, v)$ such that $v$ is the left child of $u_i$ for some $1 \le i < m$,
which can be computed by $\rank_0(\da, d_0+m-1) - \rank_0(\da, d_0)$ with $d_0 := \rank_0(\pa, u_{0})$.
Let $v_1, v_2, \dots, v_{m+1}$ be the endpoints of the non-SC-edges branching out from the SC-path 
sorted by the preorder of left-to-right traversal from $u_1$ so that $\expand{u_1} = \expand{v_1} \expand{v_2} \cdots \expand{v_{m+1}}$.
For any $1 \le i \le m+1$ it holds that
\begin{equation}\label{eq:v}
  v_i =
  \begin{cases}
    \dicone[\select_{0}(\da, \rank_{0}(\da, d_0) + i)]            & \mbox{if $i \le t$},\\
    \dictwo[2 \cdot \rank_{1}(\pa, u_{0}) + i - t]                      & \mbox{if $t < i \le t+2$},\\
    \dicone[\select_{1}(\da, \rank_{1}(\da, d_0) + m + 2 - i]    & \mbox{if $t + 2 < i$}.
  \end{cases}
\end{equation}
Hence, if we augment $\da$ with the data structure of Lemma~\ref{lem:rrr} for rank/select queries, we can compute $v_i$ in $O(1)$ time.

In order to support interval-biased search on the SC-path, 
we can use the data structure of Lemma~\ref{lem:ibs} for the sequence $(\sum_{i' = 1}^{i} |\expand{v_{i'}}|)_{i = 1}^{m+1}$ of prefix sums of $|\expand{v_i}|$'s.
However, if we store $m+1$ integers for the SC-path of $m$ nodes, the number adds up to $n+n'$ in total for all SC-paths.
An easy way to reduce this number to $n$ is to exclude the largest prefix sum $\sum_{i' = 1}^{m+1} |\expand{v_{i'}}| = |\expand{u_1}|$ for every SC-path.
$|\expand{u_1}|$ is not needed for efficient random access queries as 
we can immediately proceed to $v_{m+1}$ if a target position is greater than $\sum_{i' = 1}^{m} |\expand{v_{i'}}|$.
One drawback of not storing $|\expand{u_1}|$ explicitly is that 
the data structure cannot answer $|\expand{u_i}|$ efficiently for a given variable $u_i$
if the path from $u_1$ to $u_i$ consists only of right edges.
Since such queries might be useful in some scenarios, we show an alternative way in the following.

We decompose $\expand{u_1}$ into $m$ strings $\expand{v_1}, \dots, \expand{v_{k}}, \expand{u_m}, \expand{v_{k+3}}, \dots, \expand{v_{m+1}}$
and construct the data structure of Lemma~\ref{lem:ibs} for the sequence $g_1, g_2, \dots, g_m$ of prefix sums of these $m$ string lengths.
Now we can compute $|\expand{u_i}|$ by subtracting the prefix sum before $u_i$ begins from the prefix sum where $u_i$ ends on $\expand{u_1}$.
Since there are $\rank_{0}(\da, d_0+i-1) - \rank_{0}(\da, d_0)$ (resp. $\rank_{1}(\da, d_0+i-1) - \rank_{1}(\da, d_0)$) non-SC-edges branching out to the left (resp. right) along the path from $u_1$ to the parent of $u_i$, 
we can compute $|\expand{u_i}|$ in $O(1)$ time by $g_{i'} - g_{i''}$,
where $i' = m - (\rank_{1}(\da, d_0+i-1) - \rank_{1}(\da, d_0))$ and $i'' = \rank_{0}(\da, d_0+i-1) - \rank_{0}(\da, d_0)$.
During a search on the SC-path, if the target position falls in $\expand{u_m} = \expand{v_{t+1}} \expand{v_{t+2}}$,
we can decide which of $\expand{v_{t+1}}$ and $\expand{v_{t+2}}$ contains the target position using $|\expand{v_{t+1}}|$.
We can also compute in $O(1)$ time the prefix sum of $|\expand{v_{i'}}|$'s up to $i~(1 \le i \le m+1)$ as
\begin{equation}\label{eq:prefixsum}
  \sum_{i' = 1}^{i} |\expand{v_{i'}}| =
  \begin{cases}
    g_i                      & \mbox{if $i \le t$},\\
    g_t + |\expand{v_{t+1}}|  & \mbox{if $i = t+1$},\\
    g_{i-1}                  & \mbox{if $t + 1 < i$}.
  \end{cases}
\end{equation}

We use a global integer array $\lensum[1..n] \in [1..N]^{n}$ to store $g_1, g_2, \dots, g_m$,
in $\lensum[\select_1(\pa, r-1)+1..\select_1(\pa, r)]$.
We also consider concatenating all post-order encodings for compacted binary tries to work on a single bit string:
We store a bit string $\cbt := B_{1}B_{2} \cdots B_{n'}$ of length $2n-n'$,
where $B_{r}$ is the post-order encoding of the compacted binary trie for the $r$-th SC-path on $\pa$,
which consists of $m$ zeros (leaves) and $m-1$ ones (internal nodes).
Since every prefix of $\cbt$ has more zeros than ones,
$\cbt$ can be augmented by the succinct data structure of~\cite{2014NavarroS_FullyFunctStaticAndDynam} 
to support queries of Lemma~\ref{lem:poe} for each compacted binary trie.
In short, interval-biased search on the $r$-th SC-path can be performed using the information stored in
$B_{r} = \cbt[2 \cdot \select_1(\pa, r-1)+2-r..2 \cdot \select_1(\pa, r)+1-r] = \cbt[2u_1-r..2u_m+1-r]$ 
and $\lensum[\select_1(\pa, r-1)+1..\select_1(\pa, r)] = \lensum[u_1..u_m]$.

\begin{figure}[h!]
 \centering{
   \includegraphics[scale=0.4]{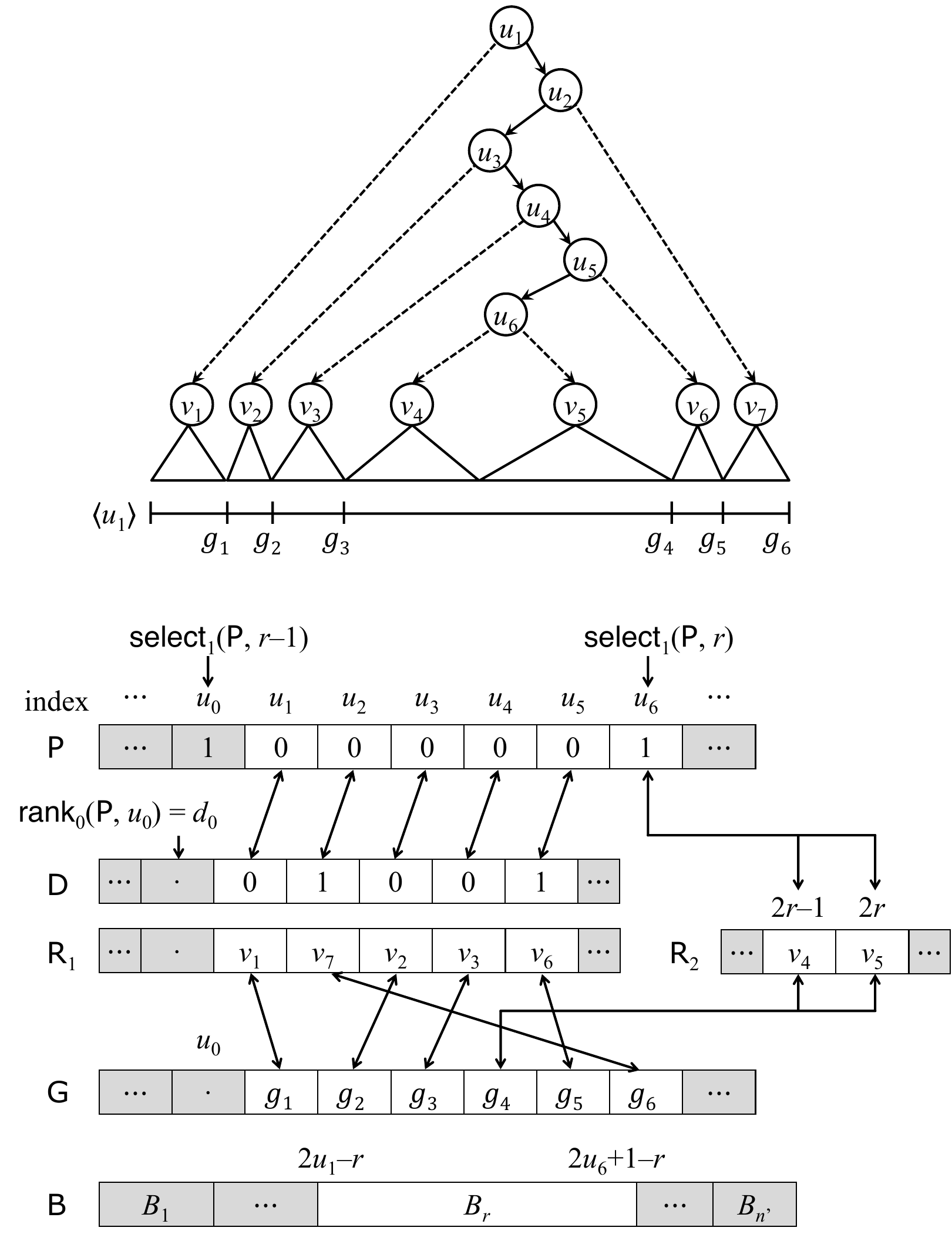}
 }
 \caption{Illustration for our encoding (I).
   Supposing that the $r$-th SC-path has $6$ nodes $(u_1, u_2, u_3, u_4, u_5, u_6)$ in the form depicted above,
   the layout of the information for this SC-path in $\pa$, $\da$, $\dicone$, $\dictwo$, $\lensum$ and $\cbt$ is shown below.
 }
 \label{fig:sc-path-enc}
\end{figure}

We are now ready to show the result (I) of Theorem~\ref{theorem:enc}:
\begin{proof}
  Encoding (I) consists of $\pa$, $\da$, $\dicone$, $\dictwo$, $\lensum$ and $\cbt$~(see Fig.~\ref{fig:sc-path-enc} for an illustration),
  and succinct data structures built on $\pa$, $\da$ and $\cbt$,
  which clearly fits in $n \lceil \lg N \rceil + (n + n') \lceil \lg (n+\sigma) \rceil + 4n - 2n' + o(n)$ bits in total.

  We follow the strategy described in Subsection~\ref{sec:overview} and use the variables introduced there in the following explanation.
  Suppose that we now arrive at $x_{i}^{\mathsf{in}}$ for some $1 \le i < e$,
  and we know that the relative target position in $\expand{x_{i}^{\mathsf{in}}}$ is $p_i$.
  Our task is to find the non-SC-edge $(x_{i}^{\mathsf{out}}, x_{i+1}^{\mathsf{in}})$
  such that $\expand{x_{i+1}^{\mathsf{in}}}$ contains $\expand{x_{i}^{\mathsf{in}}}[p_i]$ by interval-biased search 
  in $O(1 + \log |\expand{x_{i}^{\mathsf{in}}}| - \log |\expand{x_{i+1}^{\mathsf{in}}}|)$ time
  on the SC-path $(u_1, u_2, \dots, u_m) = (\select_1(\pa, r-1)+1, \select_1(\pa, r-1)+2, \dots, \select_1(\pa, r))$ with $r = \rank_1(\pa, x_{i+1}^{\mathsf{in}}-1) + 1$.
  Recall that $\lfloor \lg |\expand{u_{m}}| \rfloor = \lfloor \lg |\expand{u_{1}}| \rfloor$ by the definition of the symmetric centroid decomposition.
  Since $O(\log |\expand{x_{i}^{\mathsf{in}}}| - \log |\expand{x_{i+1}^{\mathsf{in}}}|) = O(\log |\expand{u_{1}}| - \log |\expand{x_{i+1}^{\mathsf{in}}}|)$,
  it is fine to start interval-biased search from the root of the compacted binary trie for the SC-path
  to the relative target position $p'_{i}$ in $\expand{u_1}$.
  We set $p'_{i} = p_i + \lensum[u_1+t'-1]$ if $t' > 0$, and otherwise $p'_{i} = p_i$, where
  $t'$ is the number of non-SC-edge branching out to the left from the path $(u_1, u_2, \dots, u_{j-1})$,
  which can be computed by $t' = \rank_{0}(\da, d_0 + x_{i}^{\mathsf{in}} - u_1) - \rank_{0}(\da, d_0)$ with $d_0 = \rank_0(\pa, u_{1} - 1)$.

  Let $v_1, v_2, \dots, v_{m+1}$ be the endpoints of the non-SC-edges branching out from the SC-path 
  sorted by the preorder of left-to-right traversal from $u_1$ so that $\expand{u_1} = \expand{v_1} \expand{v_2} \cdots \expand{v_{m+1}}$.
  Interval-biased search (and possibly some additional work if the target position falls in the child of $u_m$)
  finds the index $s$ with $x_{i+1}^{\mathsf{in}} = v_{s}$ from which we can compute the value of $x_{i+1}^{\mathsf{in}}$ using Equation~\ref{eq:v}.
  We also compute the relative target position $p_{i+1}$ in $\expand{x_{i+1}^{\mathsf{in}}}$
  as $p_{i+1} = p'_{i} - \sum_{i' = 1}^{s-1} |\expand{v_{i'}}|$ using Equation~\ref{eq:prefixsum}.

  In order to retrieve $q - p + 1$ consecutive characters from $p$ efficiently, 
  we push the pair of index $s$ and $z$ in a stack before moving to $v_{s}$,
  where $v_{z}$ is the rightmost variable in the subtree rooted at $x_{i}^{\mathsf{in}}$,
  which can be computed by $z = m + 1 - (\rank_{1}(\da, d_0 + x_{i}^{\mathsf{in}} - u_1) - \rank_{1}(\da, d_0))$.
  If $|\expand{v_{s}}[p_{i+1}..]| < q - p + 1$, we will be back to this SC-path
  after computing $\expand{v_{s}}[p_{i+1}..] = \gtext[p..p+|\expand{v_{s}}| - p_{i+1}]$ below $v_{s}$:
  We pop $s$ and $z$, and move to $v_{s+1}$ to expand the prefix of $\expand{v_{s+1}}$ if $s+1 \le z$,
  and otherwise go back to the previous SC-path.
  This process goes on until we get $q - p + 1$ consecutive characters.

  Note that, for any variable $x$, we can expand $\expand{x}$ in $O(|\expand{x}|)$ time
  because the derivation tree of $x$ has $O(|\expand{x}|)$ nodes.
  During the computation of $\gtext[p..q]$,
  we meet $O(\log N)$ variables that partially contribute to $\gtext[p..q]$ (such as $x_{i}^{\mathsf{in}}$'s)
  along the paths from the root to $p$ and $q$ based on SC-paths.
  Once we find these marginal paths in $O(\log N)$ time, the incompletely expanded variables can be decomposed into 
  the sequence of fully expanded variables for which the cost of expansion can be charged to $O(q - p)$.
  Hence, we can compute $\gtext[p..q]$ in $O(\log N + q - p)$ time.
\end{proof}

\subsection{Our Encoding (II) of Theorem~\ref{theorem:enc}}
Our Encoding (II) follows the basic strategy of encoding (I) but considers reducing $(n + n') \lceil \lg (n + \sigma) \rceil$ bits, 
which are used by $\dicone$ and $\dictwo$ for the information of the endpoints of the $n+n'$ non-SC-edges, 
to $n \lceil \lg (n + \sigma) \rceil + n + 3n' + o(n)$ bits.
In so doing, we arrange the SC-paths in a certain order (which is arbitrary in encoding (I)) 
so that $n'-1$ endpoints can be retrieved implicitly through a tree $\treee$ defined below.

Firstly, $\treee$ has $n'$ nodes each of which corresponds to an SC-path.
For every SC-path $Q$ that does not start with the root of $\gdag_{\gslp}$,
there is at least one non-SC-edge coming into the topmost node of $Q$.
Then we choose arbitrary one of those and add an edge $(P, Q)$ into $\treee$, 
where $P$ is the SC-path in which the starting node of the chosen non-SC-edge exists.
In total, we choose $n'-1$ non-SC-edges, and the added edges on $\treee$ ensure that 
$\treee$ becomes a tree rooted at the SC-path starting with the root of $\gdag_{\gslp}$.
In addition, if an SC-path $P$ has distinct children $Q$ and $Q'$ in $\treee$,
they are ordered by their corresponding chosen non-SC-edges sorted by the preorder of left-to-right traversal from $u_1$.

Now we consider aligning the SC-paths on $\pa$ in breadth-first order on $\treee$.
\footnote{Any other order such as preorder and post-order is fine as long as we can efficiently compute the rank under the order on a succinct representation of $\treee$ of $2n' + o(n')$ bits. In this paper we adopt breadth-first order as all operations needed on $\treee$ can be implemented simply on LOUDS.}
On the Level-Ordered Unary Degree Sequence (LOUDS) representation~\cite{1989Jacobson_SpaceEfficStaticTreesAnd_FOCS} of $\treee$ 
we can compute, given a breadth-first order rank of an SC-path $P$ and an integer $i$, the breadth-first order rank of the $i$-th child of $P$ on $\treee$
in $O(1)$ time and $2n' + o(n')$ bits of space.
Let $L$ be the integer array of length $n+n'$ such that the endpoints of the non-SC-edges branching out from the $r$-th non-SC-path are stored in $L[\select_1(\pa, r-1)+r..\select_1(\pa, r)+r]$ in left-to-right order.
Let $\marke$ be the bit string of length $n+n'$ such that $\marke[i] = 1$ if and only if $L[i]$ corresponds to the endpoint of a non-SC-edge chosen in the construction of $\treee$.
By definition, $\marke$ contains $n'-1$ ones and $n+1$ zeros.
Then let $\dice$ be the integer array of length $n+1$ such that $\dice[i] = L[\select_0(\marke, i)]$ for any $i~(1 \le i \le n+1)$.
Note that encoding (II) stores $\dice$ but not $L$.
While the endpoint $L[i]$ with $\marke[i] = 0$ is stored explicitly in $\dice$,
the endpoint $L[i]$ with $\marke[i] = 1$ is retrieved using $\treee$.

\begin{figure}[h!]
 \centering{
   \includegraphics[scale=0.4]{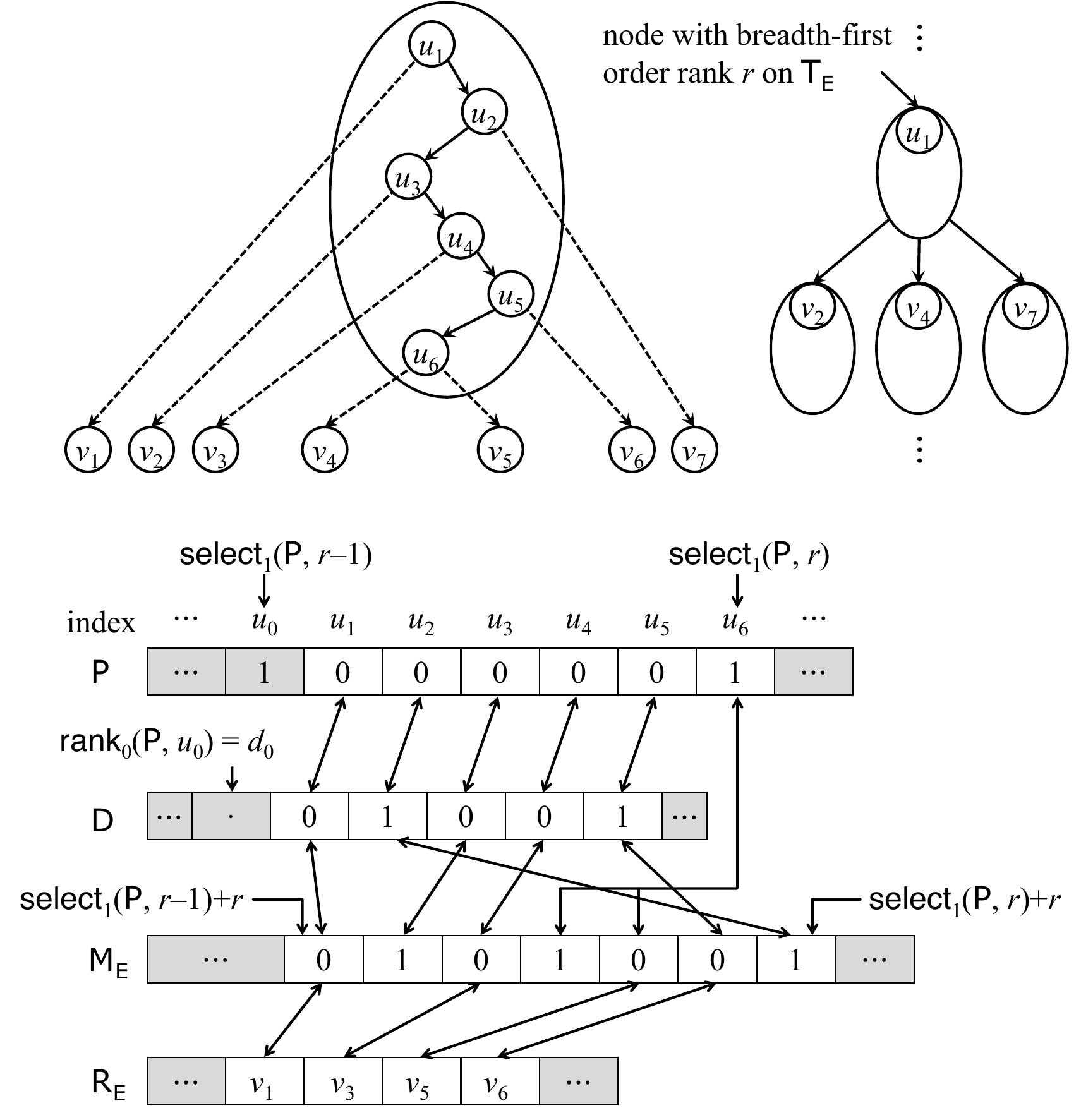}
 }
 \caption{Illustration for our encoding (II).
   Supposing that the $r$-th SC-path has $6$ nodes $(u_1, u_2, u_3, u_4, u_5, u_6)$ in the form depicted left above
   and has three child SC-paths on $\treee$ starting with $v_2, v_4$ and $v_7$,
   the layout of the information for this SC-path in $\pa$, $\da$, $\marke$ and $\dice$ is shown below.
 }
 \label{fig:sc-path-enc-2}
\end{figure}

We are now ready to show the result (II) of Theorem~\ref{theorem:enc}:
\begin{proof}
  Encoding (II) consists of $\pa$, $\da$, $\dice$, $\marke$, $\treee$, $\lensum$ and $\cbt$,
  and succinct data structures built on $\pa$, $\da$, $\marke$, $\treee$ and $\cbt$.
  Compared to encoding (I), $\dicone$ and $\dictwo$ of $(n + n') \lceil \lg (n+\sigma) \rceil$ bits 
  are replaced with $\dice$, $\marke$ and $\treee$ of $n \lceil \lg (n + \sigma) \rceil + n + 3n' + o(n)$ bits.
  Thus the space usage fits in $n \lceil \lg N \rceil + n \lceil \lg (n+\sigma) \rceil + 5n + n' + o(n)$ bits in total.

  Algorithmic difference from encoding (I) is only in how we get the endpoints of non-SC-edges.
  Let us consider the $r$-th SC-path $P$ on $\pa$.
  Note that the sorted list of endpoints of non-SC-edges for the $r$-th SC-path on $\pa$ 
  is $L[\select_{1}(\pa, r-1)+r..\select_{1}(\pa, r)+r]$, which is not stored explicitly.
  Now suppose that we want to compute $L[i]$ for $\select_{1}(\pa, r-1)+r \le i \le \select_{1}(\pa, r)+r$.
  If $\marke[i] = 0$, we get $L[i]$ directly from $\dice[\rank_0(\marke, i)]$.
  If $\marke[i] = 1$, $L[i]$ is the topmost node of the SC-path $Q$,
  where $Q$ is the $k$-th child of $P$ on $\treee$ with $k = \rank_1(\marke, i) - \rank_1(\marke, \select_{1}(\pa, r-1)+r-1)$.
  Then we compute the breadth-first order rank $r'$ of the $k$-th child $Q$ of $P$ on $\treee$
  and get $L[i]$ by $\select_{1}(\pa, r'-1) + 1$.
  See Fig.~\ref{fig:sc-path-enc-2} for an illustration to see how 
  $\dice$, $\marke$ and $\treee$ store the endpoints of non-SC-edges.

  As shown above, getting $L[i]$ takes in $O(1)$ time.
  Any other task for random access is done in the same way as encoding (I).
\end{proof}

\subsection{Our Encoding (III) of Theorem~\ref{theorem:enc}}
Encoding (II) chooses $n'-1$ non-SC-edges so that we can define a tree that connects all non-leaf SC-paths to represent their endpoints implicitly, taking $n \lceil \lg (n + \sigma) \rceil + n + 3n' + o(n)$ bits for this part.
Encoding (III) uses a simpler strategy, which takes $n \lceil \lg (n + \sigma) \rceil + n + n' + \sigma + o(n+\sigma)$ bits instead:

We fix a uniform way, that can be applied to every SC-path, to choose exactly one endpoint from the list of endpoints of the non-SC-edges branching out from an SC-path.
For example, we can decide to choose the first/last endpoint of the list or the left/right child of the last node of the SC-path.
Note that a chosen endpoint might be a terminal symbol.
Here, we consider terminal symbols are in $[1..\sigma]$ and variables (non-terminals) are in $[\sigma+1..n+\sigma]$, violating our assumption in the preliminary.\footnote{In this setting, we map an integer assigned to a variable to the integer without offset $\sigma$ when necessary.}
Then, we sort all SC-paths so that the sequence of the chosen endpoints becomes monotonically non-decreasing.
The sequence can be stored in a bit array $\monoe$ of length $\le n + n' + \sigma$ by concatenating the unary representation of the difference of each element of the sequence from the previous element with an additional bit put after every unary representation to separate them (consider the sentinel with value zero in front of the sequence).
If '1' is used for a separator, the $i$-th element of the sequence can be computed in constant time by $\select_{1}(\monoe, i) - i$.
The endpoints of other non-SC-edges are stored explicitly in an array using $n \lceil \lg (n + \sigma) \rceil$ bits.

When we want to access the endpoint of a non-SC-edge,
we first check if it is a chosen one using the information stored in $\pa$ and $\da$.
If so, we get the endpoint from $\monoe$.
If not, we get it from the array that stores it explicitly.

\section{Conclusions and Future Work}
In this paper, we presented novel space-efficient SLP encodings that support random access in almost optimal time.

We conclude with several future research directions:
\begin{itemize}
\item To reduce the most space consuming term $n \lceil \lg N \rceil$ bits of $\lensum$ in our encoding.
  For example, we may be able to think about storing each value of $g_1, g_2, \dots, g_m$ in $\lceil \lg g_m \rceil$ bits instead of $\lceil \lg N \rceil$ bits.
\item To improve extraction time for $\gtext[p..q]$ to $O(\log N + (q - p)/\log_{\sigma} N)$ like in~\cite{2015BelazzouguiCPT_AccesRankAndSelecIn_ESA}
  but in a smaller additional space.
\item To consider a space-efficient encoding that allows us to extract $\expand{x}[p..q]$ in $O(\log |\expand{x}| +q - p)$ time for a given variable $x$.
  Using the same random access procedure described in the proof of Theorem~\ref{theorem:enc},
  we can extract $\expand{x}[p..q]$ in $O(\log N + q - p)$ time.
  Unfortunately, the $\log N$ term cannot be reduced to $\log |\expand{x}|$ because
  there could still be $O(\log N)$ non-SC-edges on the path from $x$ to the target leaf.
\item To devise practical implementation based on the theoretical result of this paper.
\end{itemize}

\section*{Acknowledgements}
Tomohiro I was supported by JSPS KAKENHI (Grant Numbers 22K11907 and 24K02899) and JST AIP Acceleration Research JPMJCR24U4, Japan.
We thank Narumi Zenitsubo for the discussion we had for his graduation thesis.
We also thank anonymous reviewers for their feedback, which helps improve the manuscript greatly.

\bibliographystyle{plainurl}
\bibliography{refs}

\end{document}